\documentclass[11pt]{article}
\usepackage[hidelinks]{hyperref}
\usepackage{amsmath}
\usepackage{subcaption} 
\usepackage{amsthm}
\usepackage{fullpage}
\usepackage{amssymb}
\usepackage{graphicx}
\usepackage{algorithm}
\usepackage{algpseudocode}
\usepackage{tikz}
\usetikzlibrary{arrows}
\usepackage{titling}
\usepackage{color}

\usepackage{xcolor}
\hypersetup{
    colorlinks,
    linkcolor={red!50!black},
    citecolor={blue!50!black},
    urlcolor={blue!80!black}
}

\graphicspath{{images/}}
\newtheorem{thm}{Theorem}

\newtheorem{claim}{Claim}
\newtheorem{conjecture}{Conjecture}
\theoremstyle{definition}

\begin{document}

\renewcommand{\thefootnote}{\arabic{footnote}}

\title{Some Counterexamples for Compatible Triangulations}
\author{Cody Barnson\thanks{Department of Mathematics and Computer Science, University of Lethbridge, {\tt cody.barnson@uleth.ca}}
\and Dawn Chandler\thanks{School of Computing Science, Simon Fraser University, {\tt dschandl@sfu.ca}}
\and Qiao Chen\thanks{School of Computing Science, Simon Fraser University, {\tt qiaoc@sfu.ca}} 
\and Christina Chung\thanks{Department of Computer Science, University of Toronto, {\tt chr.chung@mail.utoronto.ca}}
\and Andrew Coccimiglio\thanks{School of Computing Science, Simon Fraser University, {\tt acoccimi@sfu.ca}}
\and Sean La\thanks{Department of Mathematics, Simon Fraser University, {\tt sean\_la@sfu.ca}} 
\and Lily Li\thanks{Department of Mathematics and the School of Computing Science, Simon Fraser University, {\tt xyl9@sfu.ca}}
\and A\"ina Linn\thanks{Department of Computer Science, McGill University, {\tt aina.georges@mail.mcgill.ca}}
\and Anna Lubiw\thanks{Cheriton School of Computer Science, University of Waterloo, {\tt alubiw@uwaterloo.ca}}
\and Clare Lyle\thanks{School of Computer Science, McGill University, {\tt clare.lyle@outlook.com}}
\and Shikha Mahajan\thanks{Cheriton School of Computer Science, University of Waterloo, {\tt s7mahaja@uwaterloo.ca}}
\and Gregory Mierzwinski\thanks{Department of Computer Science, Bishop's University, {\tt  gmierzwinski@outlook.com}} 
\and Simon Pratt\thanks{Cheriton School of Computer Science, University of Waterloo, {\tt Simon.Pratt@uwaterloo.ca}}
\and Yoon Su (Matthias) Yoo\thanks{Department of Computer Science, University of Manitoba, {\tt saiihttam@gmail.com}}
\and Hongbo Zhang\thanks{Cheriton School of Computer Science, University of Waterloo, {\tt hongbo.zhang@uwaterloo.ca}}
\and Kevin Zhang\thanks{Faculty of Applied Science, University of British Columbia, {\tt kevin@kouver.ca}}
}

\thanksmarkseries{arabic}

\maketitle

\begin{abstract}
We consider the conjecture by Aichholzer, Aurenhammer, Hurtado, and Krasser that any two points sets with the same cardinality and the same size convex hull can be triangulated in the ``same'' way, more precisely via  \emph{compatible triangulations}.
We show counterexamples to various strengthened versions of this conjecture.
\end{abstract}

\section{Introduction}

Let $P$ and $Q$ be two point sets, each with $n$ points, and with $h$ points on their convex hulls.  Triangulations $T_P$ of $P$ and $T_Q$ of $Q$ are \emph{compatible} if there is a bijection $\phi$ from $P$ to $Q$ such that $abc$ is a clockwise-ordered triangle of $T_P$ if and only if $\phi(a)\phi(b)\phi(c)$ is a clockwise-ordered triangle of $T_Q$.

A main application of compatible triangulations is the problem of morphing.
If we have compatible triangulations of two point sets in the plane then the area bounded by the convex hull of the first set can be morphed to the area bounded by the convex hull of the second set. 
In the case of very similar point sets, a linear motion of each triangle will work.  For the general case, there are more sophisticated planarity-preserving morphs~\cite{GS,morph-big}.

In 2003 Aichholzer, Aurenhammer, Hurtado, and Krasser~\cite{AICHHOLZER20033} conjectured that compatible triangulations always exist:

\begin{conjecture}
If $P$ and $Q$ are point sets in general position (i.e.~with no 3 points collinear) with the same cardinality and the same size convex hull then they have compatible triangulations.
\label{conj:compatible-exists}
\end{conjecture}

We prove that several strengthened forms of the conjecture are false.

\subsection{Definitions and Preliminaries}

A \emph{triangulation} of a finite set of points, $P$, in the plane is a maximal set of segments $pq$, $p,q \in P$ such that no two segments $ab$ and $cd$ \emph{cross}, i.e.~intersect in a point that is not a common endpoint of the segments.
We assume that the points $P$ do not all lie on one line.  Then a triangulation is a planar graph whose interior faces are all triangles~\cite{devadoss-orourke}.

From the definition we immediately get:

\begin{claim}  Consider a set of points $P$.  If $a,b \in P$ have the property that the interior of segment $ab$ is not intersected by any other segment $pq$ with $p,q \in P$, then $ab$ must be part of any triangulation of $P$.  
\label{claim:forced-edge}
\end{claim}

This claim implies that the convex hull segments are part of any triangulation.  We can also use the claim to show that certain other segments must be in any triangulation:

\begin{claim}
Let $p$ be a point of the convex hull of point set $P$.  Suppose that when we remove $p$ from the point set, the interior points $q_1, \ldots, q_t$ become convex hull vertices.  Then the segments $pq_i, i=1 \ldots t$ must be part of any triangulation of $P$.
\label{claim:forced-edges-to-CH}
\end{claim}
\begin{proof}
Each segment $pq_i$ satisfies the condition for Claim~\ref{claim:forced-edge}.
\end{proof}

\section{Collinear Points}

Aichholzer et al.~\cite{AICHHOLZER20033} gave a counterexample to show that Conjecture~\ref{conj:compatible-exists} fails if there are collinear points.  Their example has $n=7$ and is shown in Figure~\ref{fig:Aich-counterex}.  In this section we give a smaller counterexample with $n=6$.  We also give an alternate counterexample with $n=7$.  Furthermore, we show that there is no counterexample with $n \le 5$ points.  

\begin{figure}[h]
  \centering
  \includegraphics[width=0.28\linewidth,page=1]{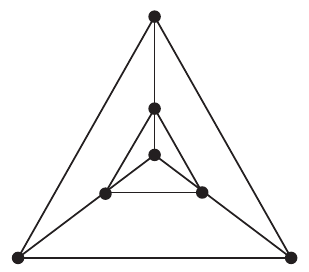}
  \ \ 
  \includegraphics[width=0.28\linewidth,page=2]{Aich-counterex.pdf}
  \caption{The counterexample to Conjecture~\ref{conj:compatible-exists} in case of collinear points from Aichholzer et al.~\cite{AICHHOLZER20033}.}
  \label{fig:Aich-counterex}
\end{figure}

\begin{figure}[h]
  \centering
  \includegraphics[width=0.28\linewidth,page=4]{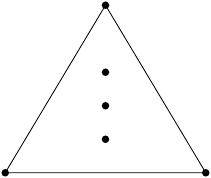}
  \ \ 
  \includegraphics[width=0.28\linewidth,page=5]{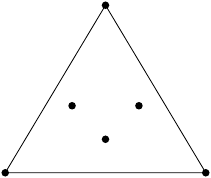}
  \caption{Point sets $P$ and $Q$.}
  \label{fig:6-point-counterexample-1}
\end{figure}

\begin{figure}[h]
  \centering
  \includegraphics[width=0.28\linewidth,page=1]{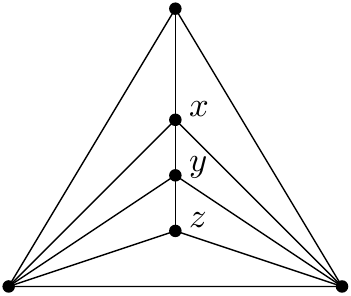}
  \includegraphics[width=0.28\linewidth,page=2]{6-point-counterexample.pdf}
  \includegraphics[width=0.28\linewidth,page=3]{6-point-counterexample.pdf}
  \caption{$T_{P}$, $T_{Q_1}$, and $T_{Q_2}$.}
  \label{fig:6-point-counterexample-2}
\end{figure}

Let $P$ and $Q$ be the left and right point sets of Figure~\ref{fig:6-point-counterexample-1}, respectively. We claim that a compatible triangulation between $P$ and $Q$ does not exist. To do this, first we will prove that the triangulations $T_P$, $T_{Q_1}$ and $T_{Q_2}$ as shown in Figure~\ref{fig:6-point-counterexample-2} are the only triangulations for their point sets. Then, we will show that $T_P$ is compatible with neither $T_{Q_1}$ nor $T_{Q_2}$.

First consider point set $P$.  Observe that every edge of triangulation $T_P$ is forced by Claim~\ref{claim:forced-edge}.

Now consider point set $Q$.  Two edges from each point on the convex hull to $a,b$ and $c$ are forced by Claim~\ref{claim:forced-edges-to-CH}.  If we add an additional edge from a convex hull point  to an interior point we get triangulation $T_{Q_2}$, and its rotationally symmetric copies. If we do not add any additional edges from convex hull points then triangulation $T_{Q_1}$ is the unique possibility. 



To show that $P$ and $Q$ are not compatible, we look at the degrees of the hull points of their triangulations. $T_P$'s hull points have degrees of 3, 5, and 5. $T_{Q_1}$'s have degrees 4, 4, and 4, and $T_{Q_2}$'s have degrees 4, 4, and 5. Therefore, the triangulation $T_P$ is compatible with neither $T_{Q_1}$ nor $T_{Q_2}$. Since these are the only triangulations, we can conclude that no triangulation exists between the point sets $P$ and $Q$.

We give another counterexample with $n=7$ in Figure~\ref{fig:alt-collinear-ex}.  

\begin{figure}[h]
  \centering
  \includegraphics[width=0.35\linewidth,page=1]{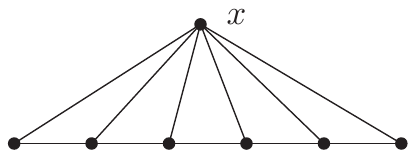}
  \ \ \ \ 
  \includegraphics[width=0.2\linewidth,page=2]{alt-collinear-ex.pdf}
  \caption{Another counterexample to Conjecture~\ref{conj:compatible-exists} for the case of collinear points.  Observe that the left hand triangulation is unique and point $x$ has degree 6, but no point on the right can have degree 6 in any triangulation.}
  \label{fig:alt-collinear-ex}
\end{figure}

\subsection{Compatible triangulations of small point sets with collinearities}

In this subsection we prove that the above counterexample with $n=6$ is the smallest possible. In other words, we prove that any point sets with $n \le 5$ points have compatible triangulations even in the presence of collinear points. 

We do this 
by proving the stronger result that any point sets with at most 2 internal points have compatible triangulations, even in the presence of collinearities.
This implies that any point sets with $n \le 5$ points have compatible triangulations because the convex hull has at least 3 points, leaving at most 2 internal points.

\begin{thm}
If two points sets of size $n$ (possibly with collinear points) have 1 or 2 internal points, then they have a compatible triangulation.
\label{thm:2-internal}
\end{thm}

Aichholzer et al.~\cite{AICHHOLZER20033} proved this for point sets with no collinearities.  
We give a proof that handles collinearities.
In fact they proved the stronger result that there is a compatible triangulation even when the convex hull mapping is forced, i.e.~when one point of the convex hull of the first point set is mapped to one point of the convex hull of the second point set.  However, that stronger result is no longer true in the presence of collinearities as shown in Figure~\ref{fig:conj-2-false}.

\begin{figure}[h]
  \centering
  \includegraphics[width=0.22\linewidth,page=1]{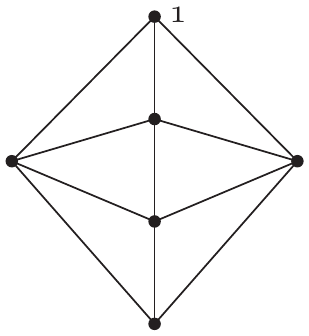}
  \ \ \ \ 
  \includegraphics[width=0.22\linewidth,page=2]{conj-2-false.pdf}
  \caption{With the indicated mapping of convex hull points, these point sets have no compatible triangulation.  Note that the triangulations are unique.}
  \label{fig:conj-2-false}
\end{figure}

\begin{figure}[ht]
  \centering
  \includegraphics[width=0.25\linewidth]{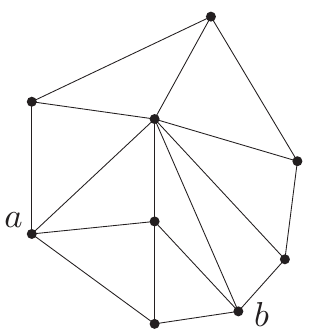}
  \caption{When $a$ and $b$ are on opposite sides of the line through the two internal points, we triangulate by joining $a$ and $b$ to the two internal points, then joining every other convex hull point to the unique internal point it sees.}
  \label{fig:2-internal}
\end{figure}

\begin{proof}[Proof of Theorem~\ref{thm:2-internal}]
For one internal point, construct compatible triangulations by adding edges from the internal point to all points on the convex hull.

Now consider two points sets $P$ and $Q$ that each have two internal points.
Following the proof of 
Aichholzer et al.~\cite{AICHHOLZER20033} we will  construct a line $\ell_P$ through the two internal points of $P$ and  a line $\ell_Q$ through the internal points of $Q$, and argue that we can map two points $a_P$ and $b_P$ of $P$ that lie on opposite sides of $\ell_P$, to two points $a_Q$ and $b_Q$ of $Q$ that lie on opposite sides of $\ell_Q$.  Then we construct compatible triangulations as shown in Figure~\ref{fig:2-internal}.

If $P$ has a convex hull point $c_P$ on the line $\ell_P$ (i.e.~collinear with the two internal points) and $Q$ also has a convex hull point $c_Q$ on the line $\ell_Q$ then we map these points to each other, let $a_P$ and $b_P$ be the two neighbours of $c_P$ on the convex hull of $P$, and let $a_Q$ and $b_Q$ be the two neighbours (in the same clockwise order) of $c_Q$ on the convex hull of $Q$.  Observe that $a_P$ and $b_P$ lie on strictly opposite sides of $\ell_P$ and similarly, $a_Q$ and $b_Q$ lies on strictly opposite sides of $\ell_Q$ in $Q$.  Thus we have found the points we need in this case.

In the remaining situations, one of $P$ or $Q$ has no convex hull points on the line through its interior points.  Suppose without loss of generality that $P$ has this property.  Orient the lines $\ell_P$ and $\ell_Q$ vertically.  Since the convex hull of $P$ has at least 3 points, and none of them lie on $\ell_P$, one side of $\ell_P$ must have at least 2 convex hull points.  Suppose this is the right-hand side. Let $q$ be the bottom point where $\ell_Q$ intersects the convex hull of $Q$. Note that $q$ may or may not be a point of $Q$.  See Figure~\ref{fig:collinear-constr}.

\begin{figure}[htb]
  \centering
  \includegraphics[width=0.5\linewidth]{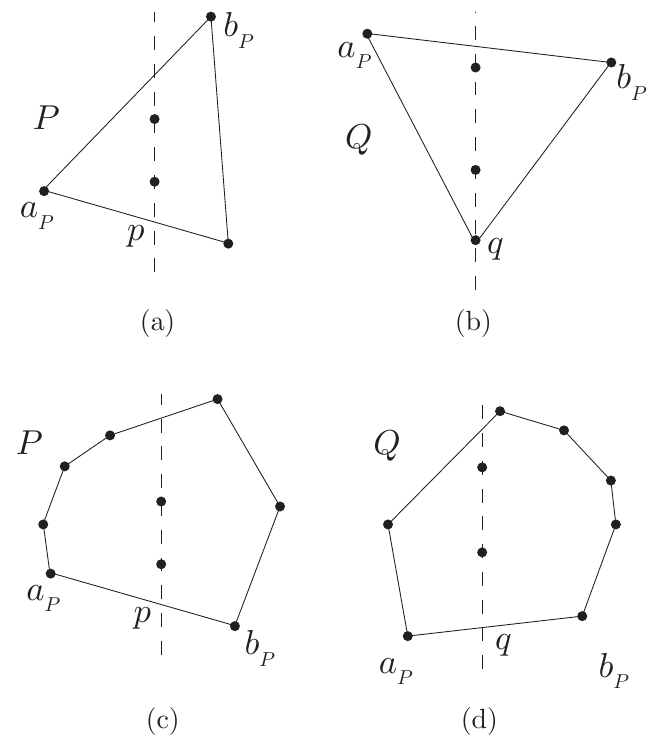}
  \caption{Finding the compatible pair $a_P$ and $a_Q$.}
  \label{fig:collinear-constr}
\end{figure}

Starting from $q$ and considering the clockwise-ordered convex hull of $Q$, let $a_Q$ be the first point of $Q$ strictly after $q$, and let $b_Q$ be the last point of $Q$ strictly before $q$.  If $q$ is not in $Q$ then $a_Q$ and $b_Q$ are adjacent on the convex hull,  and otherwise, they are separated by one convex hull point.

Now consider $P$.  Let $p$ be the bottom point where $\ell_P$ intersects the convex hull of $P$.  Note that $p$ is not in $P$.  Let $a_P$ be the first point of $P$ after $p$ in the clockwise-ordered convex hull. If $q_a$ and $q_b$ are adjacent on the convex hull of $Q$, let $b_P$ be the last point of $P$ before $p$ in the clockwise-ordered convex hull (Figure~\ref{fig:collinear-constr}(c), (d)), and otherwise, let $b_P$ be the second last point ((Figure~\ref{fig:collinear-constr}(a), (b)).  Because $P$ has at least 2 convex hull points to the left of $\ell_P$, $b_P$ is to the left of $\ell_P$.  Thus we have found the points we need. 
%
%
\end{proof}

\section{Specifying part of the mapping}

Any compatible triangulations must map convex hull points to convex hull points, and in the same clockwise order.  There is still one free choice---point $p$ on the convex hull of point set $P$ can be mapped to any point of the convex hull of $Q$.  
Aichholzer, et al.~make the stronger conjecture that compatible triangulations exist even when the mapping of $p$ to a point of the convex hull of $Q$ is fixed.

We examined what happens when the mapping is specified for more points.   
Our example in Figure~\ref{fig:fixed:points:counterexample} shows that the conjecture fails if we specify the mapping of the convex hull and of one internal point.

\begin{figure}[h]
  \centering
  \includegraphics[width=0.28\linewidth,page=1]{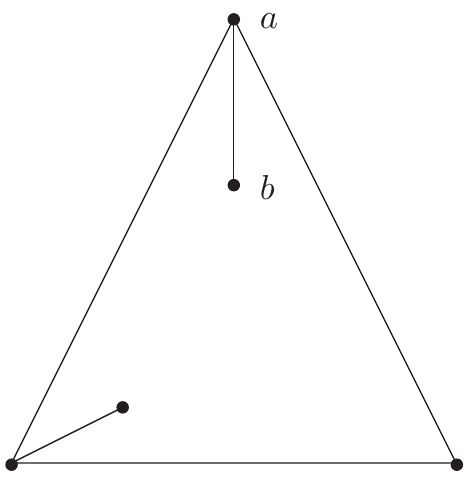}
  \ \ \ \ 
  \includegraphics[width=0.28\linewidth,page=2]{mapped-points-counterexample.pdf}
  \caption{When the mapping of points $a$ and $b$ is specified there is no compatible triangulation because the forced edges shown on the left cross on the right.}
  \label{fig:fixed:points:counterexample}
\end{figure}

\section{Specifying some of the edges}

The convex hull edges must be present in any triangulation of a point set, i.e.~they are ``forced''.    We examined what happens when we specify other edges that must be present in the triangulation.  Note that two simple polygons on $n$ points do not always have compatible triangulations~\cite{aronov1993compatible}, so if we specify $n$ edges that must be in the triangulation then there may be no compatible triangluation.  

Our example in Figure~\ref{fig:fixed:edge:counterexample}
shows that even specifying one non-convex hull edge that must be included in the triangulation leads to a counterexample.  Observe that for the point set on the left, every convex hull point is forced to have edges to two internal points by Claim~\ref{claim:forced-edges-to-CH}.  However, the required edge on the right prevents the left-hand convex hull point from having edges to two internal points.

\begin{figure}[h]
  \centering
  \includegraphics[width=0.28\linewidth,page=1]{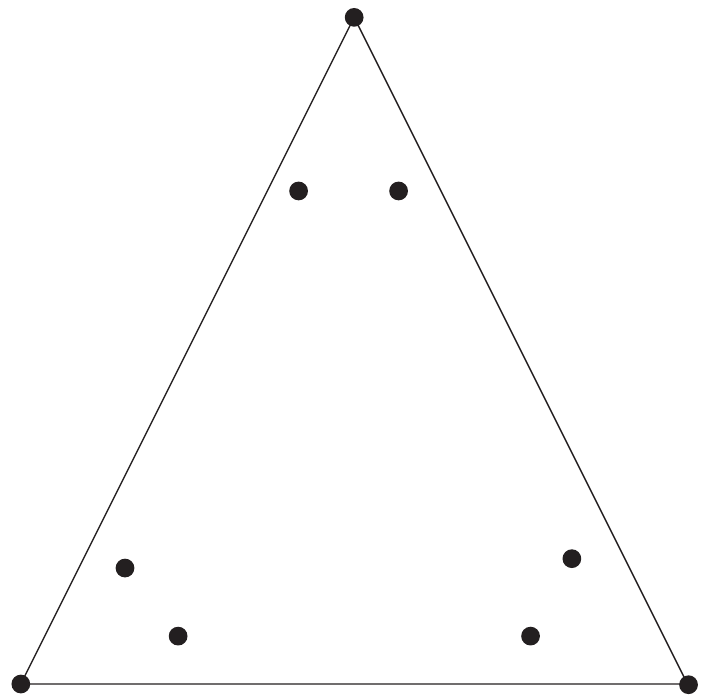}
  \ \ \ \ 
  \includegraphics[width=0.28\linewidth,page=2]{fixed-edge-counterexample.pdf}
  \caption{If we require the internal edge shown on the right, then there is no compatible triangulation between these point sets.}
  \label{fig:fixed:edge:counterexample}
\end{figure}

\section{Steiner points when the mapping is specified}

Another direction of research on compatible trianglations is to allow the addition of extra points, called \emph{Steiner points}.

Although a proof of Conjecture~\ref{conj:compatible-exists} would imply that two point sets have compatible triangulations with no Steiner points, there has been some partial progress on bounding the number of Steiner points needed.  Danciger, et al.~\cite{danciger2006compatible} showed that two Steiner points suffice if they may be placed far outside the convex hull of the points.  Aichholzer, et al.~\cite{AICHHOLZER20033} showed that a linear number of Steiner points inside the convex hull always suffice.  This is in contrast to the case of two polygons where a quadratic number of Steiner points always suffice and are sometimes necessary~\cite{aronov1993compatible}. 

In this section we explore the situation, which arises in some applications, where the mapping of the points is given as part of the input.  
In this case, compatible triangulations might not exist, as shown in Figure~\ref{fig:fixed:points:counterexample}. In order to get compatible triangulations we must add Steiner points.

\begin{figure}[h]
  \centering
  \includegraphics[width=0.4\linewidth,page=1]{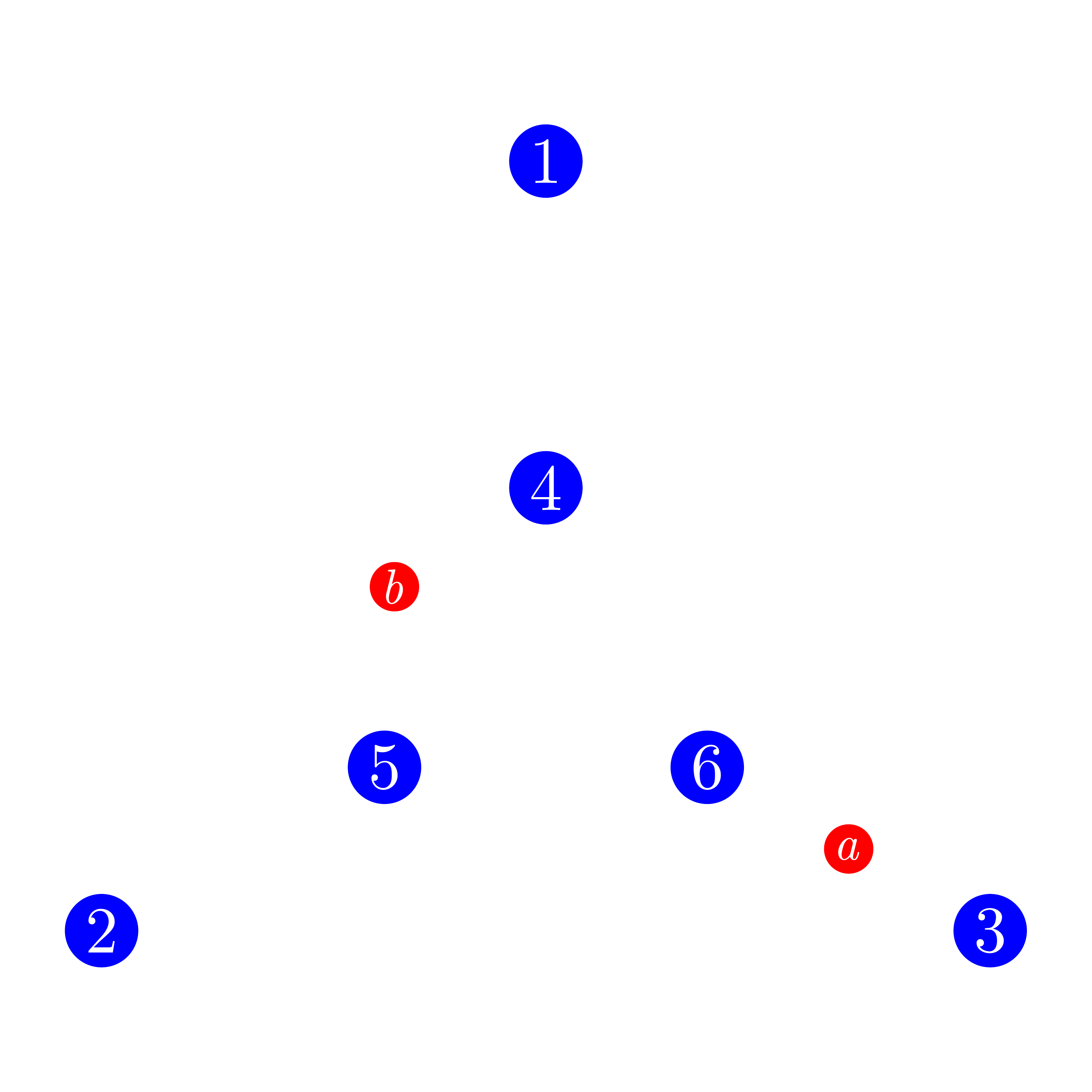}
  \ \ \ \ \ 
  \includegraphics[width=0.4\linewidth,page=2]{Steiner-points.pdf}
  \caption{The two blue point sets, with the indicated mapping of points, do not have compatible triangulations, but the addition of the two small red points allows compatible triangulations as shown in Figure~\ref{fig:Steiner-triangulations}.}
  \label{fig:Steiner-point-sets}
\end{figure}

In Figure~\ref{fig:Steiner-point-sets} we give an example of two mapped point sets with 6 points each, where compatible triangulations require the addition of two Steiner points.  A solution with two Steiner points is shown in Figure~\ref{fig:Steiner-triangulations}, and we argue below that one Steiner point is not sufficient.

It would be interesting to find a polynomial time algorithm to test if two mapped point sets have compatible triangulations with no Steiner points, or to show that the problem is NP-complete.

\begin{figure}[h]
  \centering
  \includegraphics[width=0.31\linewidth,page=3]{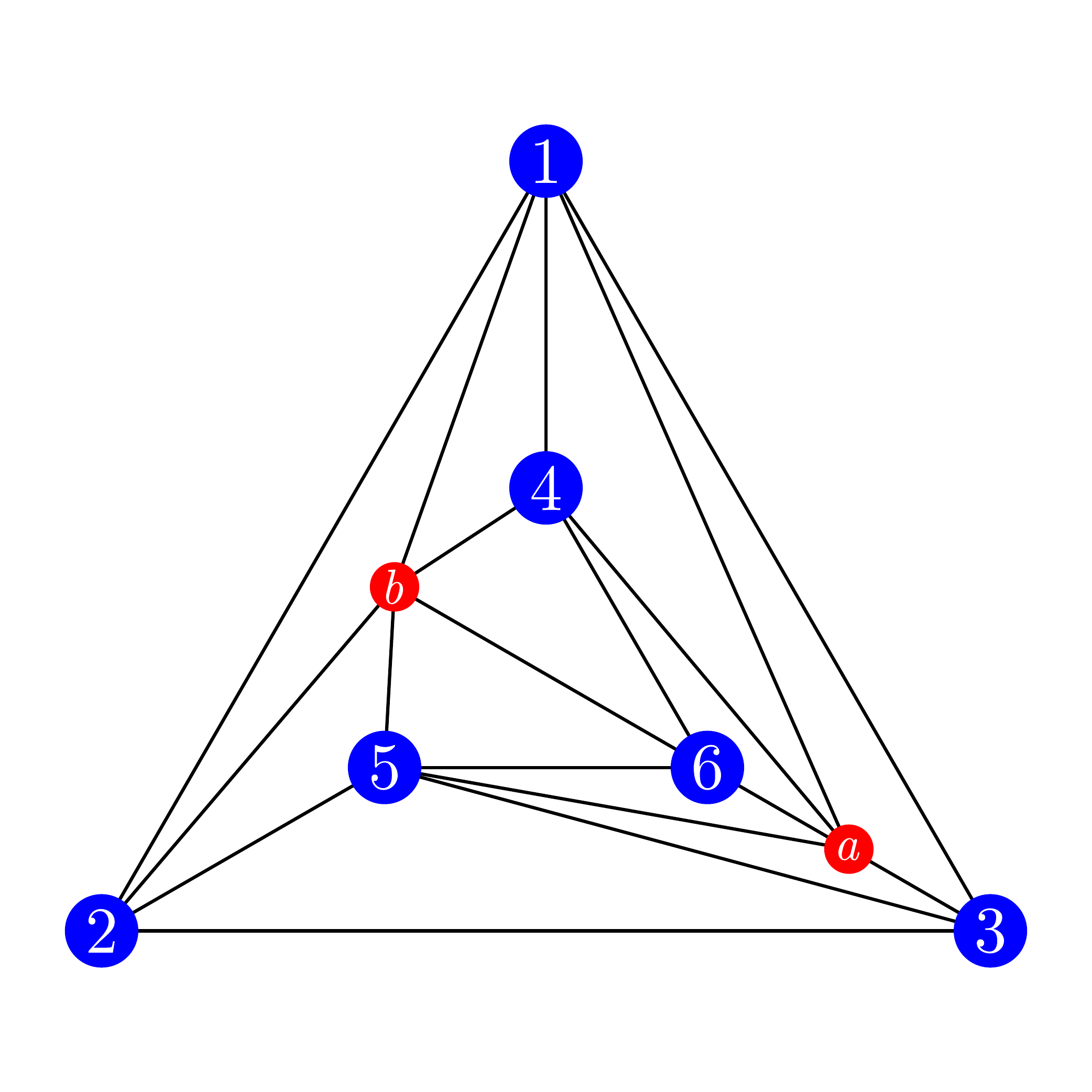}
  \ \ 
  \includegraphics[width=0.31\linewidth,page=4]{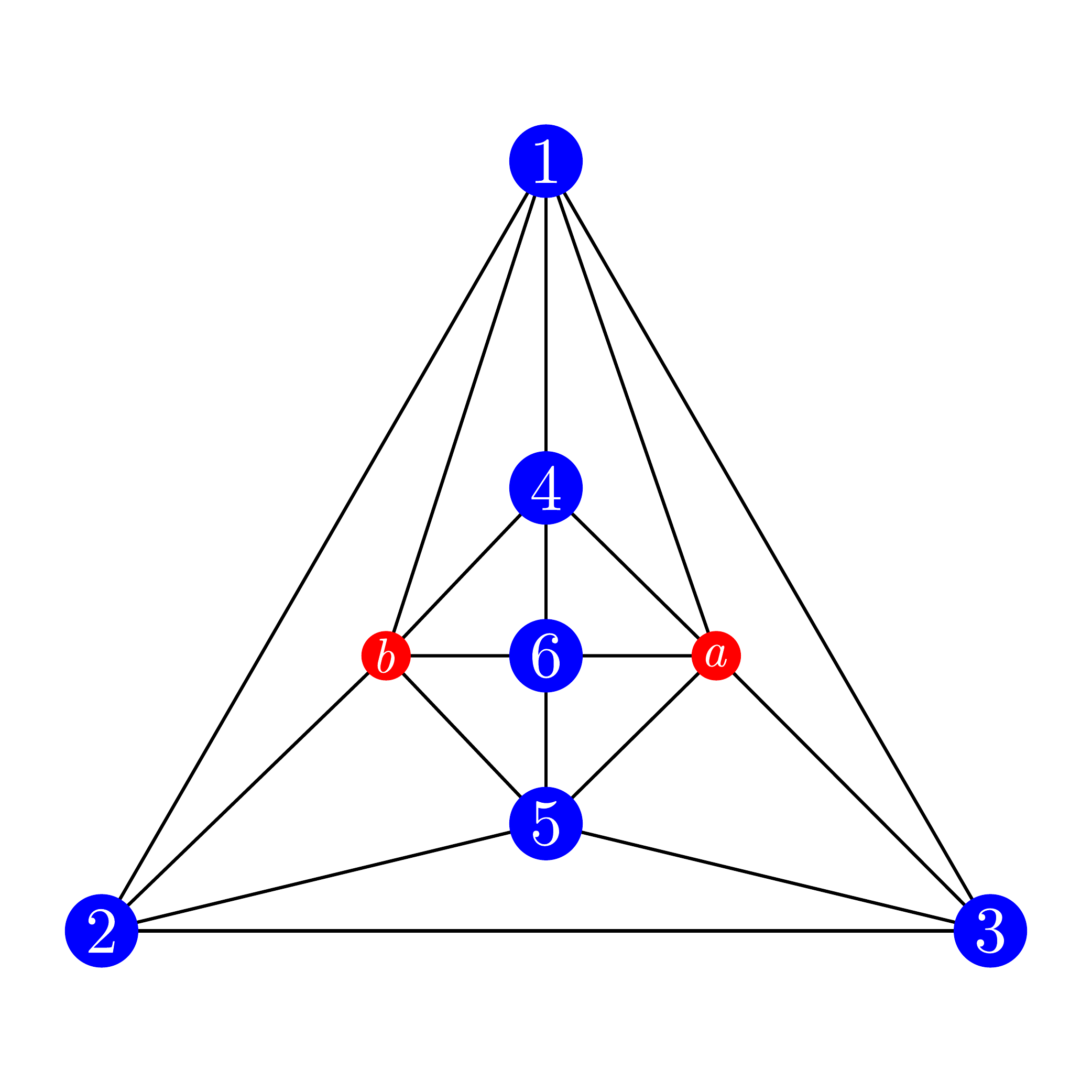}
  \ \ 
  \includegraphics[width=0.31\linewidth,page=5]{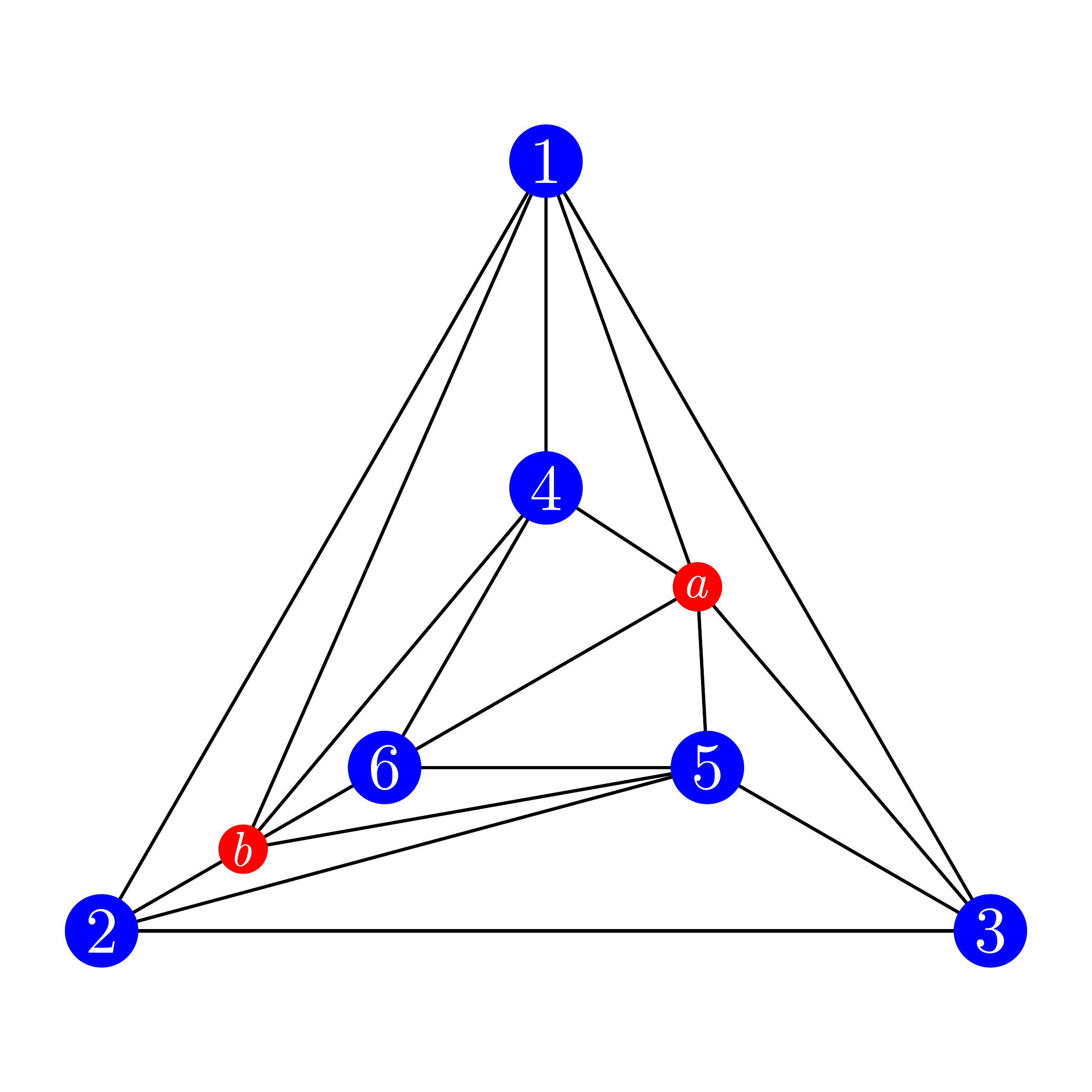}
  \caption{The two triangulations (left and right) are compatible---this can be seen most easily by comparing each one with the intermediate configuration (center).}
  \label{fig:Steiner-triangulations}
\end{figure}

\begin{claim}
The point sets shown in Figure~\ref{fig:Steiner-point-sets} do not have compatible triangulations with just one Steiner point. 
\end{claim}
\begin{proof}
Note that in compatible triangulations, each point must have the same neighbours in clockwise order. Using this property, we can show that one Steiner point is not enough to compatibly triangulate the two point sets with the given mapping. 

Claim~\ref{claim:forced-edges-to-CH} forces edges $(3,6)$, $(2,5)$ in $P$ and edges $(2,6)$, $(3,5)$ in $Q$. Observe that edges $(3,5)$ and $(2,6)$ cross in $P$ and edges $(3,6)$ and $(2,5)$ cross in $Q$. To obtain compatible triangulations of the point sets with a single Steiner point $a$, we must place $a$ so that it eliminates the crossing edges. Hence $a$ must be placed such that it breaks one of $(3,5)$ and $(2,6)$ and one of $(3,6)$ and $(2,5)$. 

First consider placing $a$ so that it breaks edge $(3,6)$ in $P$ and edge $(3,5)$ in $Q$ as shown in Figure~\ref{fig:Need-Two-Steiner-Points_Case1}. Note that in $P$, $a$ must be placed to the left of the line through 1 and 6, and below the line through 2 and 6.  Point $a$ is restricted to the analogous region in $Q$.
In $P$, by Claim~\ref{claim:forced-edges-to-CH}, we must have the edge $(2,5)$.  Similarly, in $Q$ we must have the edge $(2,6)$.  Thus these edges must be in both triangulations.  However, the cyclic order of the two edges around point $2$ is different in $P$ than in $Q$. 
Thus there can be no compatible triangulation of the point sets with the given mapping in this case. We use a similar argument for the case when $a$ breaks edges $(2,5)$ and $(2,6)$. 
\begin{figure}[h]
	\center
  \includegraphics[width=0.6\linewidth]{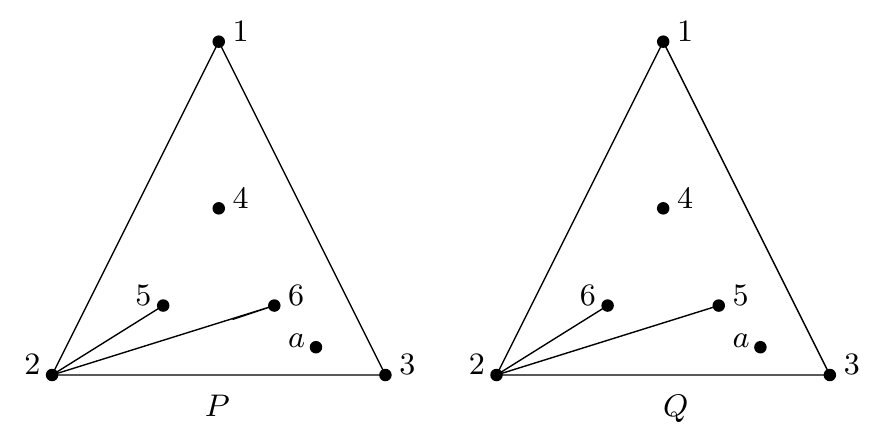}
  \caption{Edges $(2,5)$ and $(2,6)$ appear in different order in the cyclic ordering of edges around point $2$ in $P$ and $Q$.}
  \label{fig:Need-Two-Steiner-Points_Case1}
\end{figure}

Next consider  placing $a$ so that it breaks edge $(3,6)$ in $P$ and edge $(2,6)$ in $Q$. 
Claim~\ref{claim:forced-edges-to-CH} forces edges $(a,3)$ and $(a,6)$ in $P$. To break edge $(2,6)$ in $Q$, $a$ must be placed below the line through $(3,6)$, hence $a$ comes before point $5$ in the clockwise ordering of neighbours of $3$ in $Q$. To obtain the same ordering in $P$, we need $a$ below the line $(3,5)$, but then edge $(a,6)$ crosses edge $(3,5)$ as $6$ is above $(3,5)$ as shown in Figure~\ref{fig:Need-Two-Steiner-Points_Case2}. Therefore we still cannot compatibly triangulate the point sets with the given mapping. 
Finally, the case where we place $a$ to break edge $(2,5)$ in $P$ and edge $(3,5)$ in $Q$ is symmetric. 
\begin{figure}[h]
	\center
  \includegraphics[width=0.6\linewidth]{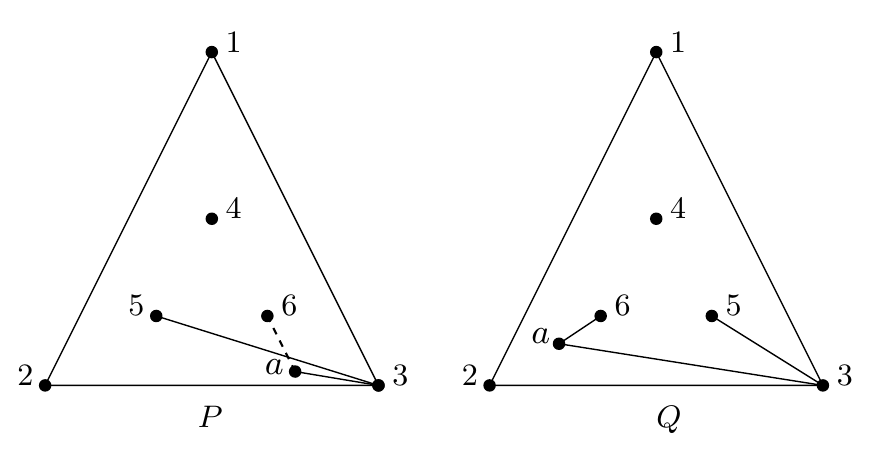}
  \caption{To obtain the same cyclic ordering of edges around point $3$ in $P$ and $Q$, we need Steiner point $a$ below $(3,5)$ in $P$, causing $(a,6)$ to cross $(3,5)$}
  \label{fig:Need-Two-Steiner-Points_Case2}
\end{figure}

\end{proof}

\section*{Acknowledgements}
This research was accomplished as part of a two day undergraduate research workshop at the University of Waterloo in October 2016.  

\bibliographystyle{plain}
\bibliography{references}

\end{document}